\documentclass[11pt]{article}
\usepackage{amssymb,amsthm,amsmath,mathrsfs}
\usepackage{prettyref,cite}
\usepackage[colorlinks,linkcolor=blue,citecolor=magenta]{hyperref}
\usepackage{graphicx}
\usepackage{algorithm,algpseudocode}
\usepackage{caption}
\usepackage [a4paper, footskip=1cm, headheight = 16pt, top=3cm, bottom=3cm,  right=2.5cm,  left=2.5cm ]{geometry}
\newcommand\pref[1]{\prettyref{#1}}
\newrefformat{chp}{\chaptername~\ref{#1}}
\newrefformat{sec}{Section~\ref{#1}}
\newrefformat{sub}{Subsection~\ref{#1}}
\newrefformat{app}{Appendix~\ref{#1}}
\newrefformat{sub}{Subsection~\ref{#1}}
\newrefformat{thm}{Theorem~\ref{#1}}
\newrefformat{qst}{Question\ref{#1}}
\newrefformat{lem}{Lemma~\ref{#1}}
\newrefformat{alg}{Algorithm~\ref{#1}}
\newrefformat{cor}{Corollary~\ref{#1}}
\newrefformat{con}{Conjecture~\ref{#1}}
\newrefformat{dfn}{Definition~\ref{#1}}
\newrefformat{rem}{Remark~\ref{#1}}
\newrefformat{fig}{\figurename~\ref{#1}}
\newrefformat{tbl}{\tablename~\ref{#1}}
\newrefformat{eq}{$($\ref{#1}$)$}
\newrefformat{fac}{Fact~\ref{#1}}

\newtheorem{pro}{Proposition}

\newtheorem{lem}[pro]{Lemma}

\newtheorem{prerem}[pro]{Remark}
\newtheorem{predefin}[pro]{Definition}

\newenvironment{dfn}{\begin{predefin}\rm}{\hfill $\blacktriangle$\end{predefin}}
\newenvironment{probc}[3]{\vskip 5pt
	\noindent {\bf #1}\\[3pt]
	\hspace*{-6pt}\begin{tabular}{p{60pt}l}
		{INSTANCE:}& \parbox[t]{12cm}{#2}\\[15pt]
		QUERY:    & \parbox[t]{12cm}{#3}}{\end{tabular}\\[5pt]}
\newcommand{\isdef}{:=}
\begin{document}

\title{\textbf{Multi-way sparsest cut problem on trees with a control on the number of parts and outliers
}}
\author{
	Ramin Javadi%
\thanks{Corresponding author, Department of Mathematical Sciences,
		Isfahan University of Technology,
		P.O. Box: 84156-83111, Isfahan, Iran. School of Mathematics, Institute for Research in Fundamental Sciences (IPM), P.O. Box: 19395-5746,
		Tehran, Iran.  Email Address: \href{mailto:rjavadi@cc.iut.ac.ir}{rjavadi@cc.iut.ac.ir}.}
\thanks{This research was in part supported by a grant from IPM (No. ...).}
\and 
	Saleh Ashkboos%
	\thanks{Department of Computer Engineering,
		Isfahan University of Technology,
		P.O. Box: 84156-83111, Isfahan, Iran. Email Address: \href{mailto:s.ashkboos@ec.iut.ac.ir}{s.ashkboos@ec.iut.ac.ir}.} 
}
\date{}
\maketitle
\begin{abstract}
Given a graph, the sparsest cut problem asks for a subset of vertices whose edge expansion (the normalized cut given by the subset) is minimized. In this paper, we study a generalization of this problem seeking for $ k $ disjoint subsets of vertices (clusters) whose all edge expansions are small and furthermore, the number of vertices remained in the exterior of the subsets (outliers) is also small. We prove that although this problem is $ NP-$hard for trees, it can be solved in polynomial time for all weighted trees, provided that we restrict the search space to subsets which induce connected subgraphs. The proposed algorithm is based on dynamic programming and runs in the worst case in $ O(k^2 n^3) $, when $ n $ is the number of vertices and $ k $ is the number of clusters. It also runs in linear time when the number of clusters and the number of outliers is bounded by a constant.  
\\
\begin{itemize}
\item[]{{\footnotesize {\bf Key words:}\ sparsest cut problem, isoperimetric number, Cheeger constant, normalized cut, graph partitioning, computational complexity, weighted trees.}}
\item[]{ {\footnotesize {\bf Subject classification:} 05C85, 68Q25, 68R10.}}
\end{itemize}
\end{abstract}
\section{Introduction}

{\large D}ata clustering is definitely among the main topics of 
modern computer science with an indispensable role in data mining, 
image and signal processing, network and data analysis, and data summarization 
(e.g. see \cite{DataClusterReview} and references therein). Considering the current status of data science, one may name some fundamental challenges in this field, among many others,  as follows:
\begin{itemize}
	\item Clustering huge and usually high-dimensional data.
	\item Clustering in presence of outliers and anomalies.
	\item Clustering non-geometric (usually non-Euclidean) data.
	\item Clustering with no prior information about the number of clusters or other features of data (as model of the source etc.).
\end{itemize}

Needless to say, in each case, efficiency and time-complexity of the proposed algorithms are global parameters with a decisive role in applicability.

The subject of this article falls into the setup of  clustering in an unsupervised and static graph-based data presentation. It is instructive to note that the graph-based approach essentially provides data presentation in a very general (not necessarily Euclidean) setting in terms of similarity kernels. In this respect, one of the main well-studied criteria is the ``\textit{sparsest cut problem}'' which apart from tremendous real-world applications in the context of spectral clustering (see e.g. \cite{shi2000normalized,ng2002spectral}), has played a crucial role in the  development of many subjects in theoretical computer science (see e.g. \cite{VaziraniBook,ChungBook}).

Our main objective in this article is to improve this approach, which is essentially based on solving a suitable subpartitioning problem on a corresponding minimum spanning tree, by providing an algorithm that not only gives rise to a fast clustering procedure, but also provides good control on determining the number of clusters and outliers. The procedure is based on a dynamic programming which runs in the worst case in $ O(k^2 n^3) $, where $ n $ is the data size and $ k $ is the number of clusters. Also, the algorithm runs in linear time in terms of the data size when the number of clusters $ k $ and the upper bound on the number of outliers are both constant (which is the case in the most prevalent applications). To the best of our knowledge, the partitioning problem solved by the proposed algorithm (\pref{alg:main}) is among the most challenging problems in this literature which is efficiently solvable, while we will also dwell on some important consequences in what follows.

\subsection{A formal setup and the main result}

Partitioning problems are essentially as old as graph theory itself, with wide applications in science and technology. In particular, one may refer to the unnormalized 
partitioning problems that usually are considered as different versions of minimum cut problems as well as the normalized versions which are more plausible in real applications, however, are much harder to resolve. One of the main problems in the category of normalized cut criteria is the \textit{sparsest cut problem} which is defined as follows. Given a graph $ G $, the sparsest cut problem asks for a cut (a subset of vertices) which has the minimum edge expansion, i.e. 
\begin{equation} \label{eq:cheeger}
\phi(G)\isdef \min_{S\subsetneq V(G), S\neq \emptyset} \max\left\{\frac{|\partial S|}{|S|}, \frac{|\partial S|}{|\overline{S}|}\right\},
\end{equation}
where $ \overline{S}\isdef V(G)\setminus S $ and $ \partial S $ is the set of all edges with exactly one end in $ S $. The sparsest cut problem is known to be an $NP-$hard problem on general graphs \cite{shi2000normalized,mohar}. 
Efforts to find an efficient algorithm for a good approximation of this problem have triggered off the development of many subfields of computer science and have had a significant influence on algorithm design and complexity theory. 
It is amazing to see that recent advances in computer science have given rise to a culmination of ideas not only from the classical graph theoretic point of view but also from the more geometric point of view discussed in the 
theory of Riemannian manifolds and stochastic processes \cite{RiemannianManifold}. 
Up to now, the best known approximation result for the sparsest cut problem is due to Arora, Rao, and Vazirani \cite{ARV04} which gives an  $O(\sqrt{\log n})$ approximation algorithm.

It is also worth noting that the invariant defined in \eqref{eq:cheeger} has an intimate connection with the second eigenvalue of the associated Laplacian operator.
In fact, relaxation of  the minimization problem in \eqref{eq:cheeger} to the Euclidean norm for real functions (i.e. changing the edge expansion 
to the Euclidean $2$-norm of the gradient of real functions which is the energy representable by the Laplacian operator) gives rise to an eigenvalue problem which is efficiently solvable,
while estimating the approximation ratio of this relaxation has led to some fundamental contributions (e.g. see \cite{AlonCheeger,AlonMilmanCheeger}).
 These relations, known as \textit{Cheeger's inequalities}, also exert considerable influence over constructing the expander graphs as well as the study of the mixing time of Markov chains (see e.g. \cite{ApprxCountingMarkov,WidgersonLectureNote}). In general, although the motivating problems in these fields of study are usually different, the synergistic effect of methods and techniques have flourished into one of the most active and productive topics in mathematics and computer science.

Recently, some generalizations of the sparsest cut problem have been studied in the literature. Here, we study a generalization which extends two-way partitioning into $k -$way connected subpartitioning and allows some vertices to lie outside the parts. 

To formulate the problem precisely, let us first fix our notation and terminology. We assume that the data is given as a simple and finite  {\it weighted graph} $G=(V,E,\omega,c)$ in which $\omega:V\to {\mathbb Q}^+$ and $c:E\to {\mathbb Q}^+$ are the vertex and edge weight functions, respectively. Note that 
in the literature close to applications the function $c$ is sometimes referred to as the {\it kernel} or the {\it similarity}, while from a geometric point of view the graph can also be considered as a discrete {\it metric-measure space}, where the distance function is usually chosen to be proportional to some inverse function of $c$. In this setting, 
by an {\it unweighted graph} we mean a graph in which all the vertex and edge weights are equal to $1$.  

Given a graph $G=(V,E,\omega,c)$ and a subset of vertices $ S\subseteq V $, the \textit{edge exapnsion} or the \textit{conductance} of $ S $, is defined as
\[\phi_G(S)\isdef \frac{c(\partial S)}{\omega(S)}, \]
where,
\[\omega(S)\isdef \sum_{u\in S} \omega(u),\quad c(\partial S)\isdef \sum_{
	e\in E(S,\overline{S})} c(e). \]
From a geometric point of view, the conductance can be interpreted as 
a {\it normalized norm of a gradient function} or a {\it normalized energy} (e.g. see \cite{GeometricSpectraRiemann,EigenRemannianGeometry} for more on the geometric interpretations).
The set ${\mathscr D}_{k}(V)$ is defined to be the set of all {\it $k$-subpartitions} $\{A_{{1}},\ldots,A_{{k}}\} \isdef \{A_{{i}}\}^k_1$ of $V$, in which $A_{{i}}$'s are nonempty disjoint subsets of $V$. 
The {\it residue} of a subpartition $\{A_{{1}},\ldots,A_{{k}}\} $ is defined to be the set $R\isdef V-\cup_{{i=1}}^{{k}} A_{{i}}$.
The set of all {\it $k$-partitions} of $V$, which is denoted by ${\mathscr
	P}_{{k}}(V)$, is the subclass of ${\mathscr D}_{{k}}(V)$ containing
all $k$-subpartitions $\{A_{{i}}\}^k_1$ for which $\cup_{{i=1}}^{{k}} A_{{i}}=V$ (i.e. $R=\emptyset$).
A subpartition (or a partition in particular) is said to be {\it connected}
if the subgraph induced on each of its parts is a connected subgraph of $ G $.
A generalization of the sparsest cut problem can be formulated as follows. 
\begin{dfn}{\label{DEFISO}
Given a weighted graph $G=(V,E,\omega,c)$ and a positive integer $k$, $1\leq k\leq |V|$,
the $ k $th {\it isoperimetric number} is defined as,
		\vspace{-5pt}
		\begin{eqnarray*}
			\iota_{{k}}(G) &\isdef&
			\displaystyle{\min_{{ \{A_{{i}}\}^{{k}}_{{1}} \in  {\mathscr
							D}_{{k}}(V) }} } \ \max_{1\leq i\leq k}\
		\phi_G(A_i).
		\end{eqnarray*}
		Furthermore, considering the partitions, the $ k $th {\it minimum normalized cut number} is defined as,
		\vspace{-5pt}
		\begin{eqnarray*}
			\tilde{\iota}_{{k}}(G) &\isdef& \displaystyle{\min_{{
						\{A_{{i}}\}^{{k}}_{{1}} \in  {\mathscr P}_{{k}}(V) }} } \
			\max_{1\leq i\leq k}\
			\phi_G(A_i).
		\end{eqnarray*}
		A vertex $v\in V$ is called a {\it $k$-outlier}, if there exists a minimizing subpartition achieving $\iota_k(G)$, while $v$ lies in its residue. 
		It is well-known that $\iota_2=\tilde{\iota}_2$  (see \cite{JCTB}) and the common value is usually called the {\it Cheeger constant} or {\it edge expansion} in the literature.
	}\end{dfn}
In this regard, Louis et al. in \cite{louis} provide a polynomial time approximation algorithm which outputs a $(1-\epsilon)k$-partition of
the vertex set such that each piece has expansion at most $O_\epsilon(\sqrt{\log n \log k})$ times $\tilde{\iota}_{{k}}(G)$ (for every positive number $ \epsilon $). Also, in \cite{oveis}, higher-order Cheeger's inequalities have been proved which relate the above parameters to the eigenvalues of the associated Laplacian Matrix (see also \cite{JCTB,Miclo}).  

Prior to formulating our problem, let us discuss some facts. 
First, one may note that as an imprecise rule of thumb, changing the cost function of a partitioning problem, 
from the normalized form to the unnormalized form, from partitions to subpartitions, or from the mean (i.e. $1$-norm) to the max (i.e. $\infty$-norm) generally makes the problem more tractable in the sense that finding more efficient algorithms to solve the problem become more probable. One of our major observations in this article is the fact that the restriction of the search space to ``connected'' subpartitions reduces the complexity of the problem too. In particular, this distinction is much comprehensible when the graph is a tree where the restriction on  subpartitions to be connected reduces the complexity of the problem from $ NP- $hard to polynomial time. Also, note that this restriction is to the best of our advantage in the sense that a cluster is more expected to be represented by a connected subgraph than a disconnected one (based on intra-similarity of the objects within a cluster). Hence,  as far as clustering is concerned, this can be considered as an acceptable assumption. 
As a matter of fact, in what follows, we show that such a change to the better will give rise to an efficient algorithm for clustering with a control on the number of parts and outliers.

We denote the main problem, i.e. the multi-way sparsest cut problem with a control on the residue number, by the acronym ``MSC problem'' which is defined as follows. 
\begin{probc}
	{MSC Problem.}
	{A weighted graph $G=(V,E,\omega,c)$, nonnegative integers $\kappa\in \mathbb{Z}^+$ and $\lambda\in \mathbb{Z}^+$ and a positive rational number $\xi\in {\mathbb Q}^+$.}
	{Does there exist a $\kappa$-subpartition of $V$ such as  $\{A_{{i}}\}^\kappa_1 \in \mathscr{D}_\kappa(V)$ such that
		$\displaystyle{\max_{1 \leq i \leq \kappa}}\{\phi_G(A_{{i}})\} \leq \xi$ and its residue number is at most $ \lambda $, i.e. $ |V\setminus \cup_{i=1}^\kappa A_i|\leq \lambda $?}
\end{probc}
\vspace{-2pt}

The MSC problem is known to be a hard problem even when the graph is of its simplest form, i.e. a tree. 
When the graph $ G $ is a tree, it is proved in \cite{dam} that MSC problem is $ NP-$complete even when the tree is unweighted and $ \lambda $ is constant (e.g. $ \lambda=0 $). Nonetheless, it is shown there that the problem is solvable in linear time for weighted trees when we drop the restriction on the residue number (i.e. $ \lambda=|V| $). An improvement of this result has effectively been applied to real clustering problems for large data-sets \cite{pr}.

The main contribution of this article (Algorithm~\ref{alg:main}) is to show that although MSC problem is $ NP-$complete for trees, it becomes tractable when the search space is restricted to connected subpartitions. In other words,  the following problem abbreviated by CMSC can be solved in polynomial time for weighted trees.
\begin{probc}
	{CMSC Problem.}
	{A weighted graph $G=(V,E,\omega,c)$, nonnegative integers $\kappa\in \mathbb{Z}^+$ and $\lambda\in \mathbb{Z}^+$ and a positive rational number $\xi\in {\mathbb Q}^+$.}
{Does there exist a connected $\kappa$-subpartition of $V$ such as $\{A_{{i}}\}^\kappa_1 \in \mathscr{D}_\kappa(V)$ such that
	$\displaystyle{\max_{1 \leq i \leq \kappa}}\{\phi_G(A_{{i}})\} \leq \xi$ and its residue number is at most $ \lambda $, i.e. $ |V\setminus \cup_{i=1}^\kappa A_i|\leq \lambda $?}
\end{probc}
\vspace{-2pt}




This result along with the fact that the minimum spanning tree of a geometric metric-measure space inherits a large part of the geometry of the space, can lead to a good approximation for MSC problem for general graphs. This can justify the importance of the problem on weighted trees when applications are concerned. 
Let us consider some consequences of this result. 

Firstly, note that given a weighted tree $ T $ and integers $ \kappa $ and $ \lambda $, finding the minimum number $ \xi $ for which there exists a connected $ \kappa- $subpartition with the residue number at most $\lambda $ and $ \max_{1 \leq i \leq \kappa}\{\phi_G(A_i)\} \leq \xi $ (as well as finding the minimizing subpartition) can be done in polynomial time by applying our algorithm iteratively along with a simple binary search. 

Secondly, given a weighted tree $ T $ and numbers $ \xi, \lambda $ (the worst edge expansion of the clusters), we can obtain a number $k_{{max}}(T,\xi)$, denoting the maximum number of parts for which the answer to CMSC problem is positive. This by itself is an important piece of information when one considers the large existing literature discussing how to determine the number of clusters for a clustering algorithm (e.g. see \cite{ExtendedKmeans} for $k$-means). 

Thirdly, from another point of view, CMSC problem can be considered as a problem of outlier-robust clustering where a solution will provide information on the number of outliers. 
It is well-known that detection of outliers and anomalies in data-sets are among the most challenging problems in the field, not just because of the hardness of the problem itself, but since the concepts themselves are quite fuzzy and depend on many different parameters as scaling or distribution of the source (e.g. see \cite{FuzzyBook, RobustReview} for the background). These facts, and in particular, lack of a universal sound and precise definition, is among the first obstacles when one is dealing with these kinds of  problems. 
In \cite{pr} some evidence has been discussed that how the data remained in the exterior of the clusters in MSC problem can be justified to be actual outliers in some sense.

Finally, our method can be extended to handle some more general semi-supervised settings where a number of training samples are given by the user which are forced or forbidden to lie in outliers (see  Section~\ref{sec:extensions}).

The organization of forthcoming sections is as follows. In Section~\ref{sec:pre}, we give required definitions and notations as well as the lemmas which justify our algorithm. In Section~\ref{sec:alg}, we present the main algorithm and explain how it can find the optimal subpartition. We also compute the time complexity of our algorithm. Finally, in Section~\ref{sec:extensions}, we discuss some extensions which handle more realistic models. 
\section{Preliminaries} \label{sec:pre}
Let $T$ be a rooted tree with root $r$. There is a natural partial order induced 
through the root on the vertices and edges of $T$ defined as $u \leq v$ for two vertices $u$ and $v$ whenever there is a path $P(r,v,u)$ in $T$ starting from $ r $ and ending at $u$ which contains $v$. Similarly,
$e \leq e'$ for two edges $e$ and $e'$ whenever there is a path $P(r,e',e)$ in $T$ starting from $ r $ and containing $e$ and $e'$ such that $e'$ is closer than $ e $ to $ r $ on $P$. In this setting, note that for any $u \not = r$ there exists a unique minimal vertex $ v $, with 
$v \geq u$ and an edge $e_{{u}} \isdef uv$, where $ v $ and $ e_u $ are called the \textit{parent vertex} and the \textit{parent edge} of $ u $, respectively (and also $ u $ is called the \textit{child} of $ v $). Also, for a given edge $e=uv$ with $u \leq v$ we may refer to $e^-=u$ and $e^+=v$, intermittently. For some technical reasons, we add one new vertex $ r' $ to $ T $ and connect it to $ r $ and define the parent edge of $ r $, $ e_r$, as the edge $rr' $. Also, we set $ \omega(r')=c(e_r)=0 $.  

If $F$ is a subset of edges of $T$, then $M(F)$ is the set of maximal elements of $F$ with respect to the natural partial order of $T$. Given a vertex $ u $ with the parent edge $ e_u $, the subtree $T_u=T_{e_u}$ refers to the subtree induced on the set $ \{v\in V(T): v\leq u\} $. Therefore, $ T_r=T_{e_r}=T $.

Let $ T=(V,E,\omega,c) $ be a weighted tree and $ \xi $ be a fixed positive number. For every integer $0\leq k\leq |V| $, define $\mathscr{C}_{{k}}(T)$ to be the class of all $k$-subpartitions $\mathcal{A}= \{A_i\}_1^k $ such that for each $ 1\leq i\leq k $, $ A_i\subseteq V(T) $ and the induced subgraph of $ T $ on $ A_i $ is connected (i.e. $ A_i $ is a subtree of $ T $). Also, given a subpartition $\mathcal{A}=\{A_i\}_1^k \in \mathscr{C}_{k}(T)$, its \textit{residue set} is defined as $ R(\mathcal{A},T)  \isdef V(T)\setminus \cup_{i=1}^{k} A_{i}$.
We also define,
$$ \phi_T(A_i) \isdef \frac{c(\partial A_i)}{\omega(A_i)}, \ \ 
\phi_T(\mathcal{A}) \isdef \displaystyle{ \max_{1\leq i\leq k}}\ \phi_T(A_{{i}})
\ \ {\rm and} \ \ 
\iota^{C}_{{k}}(T) \isdef \displaystyle{\min_{{ \mathcal{A} \in  \mathscr{C}_{{k}}(T) }}} \phi_T(\mathcal{A}).$$



In the following we describe the idea that our algorithm is based on and also prove the correctness of the algorithm. 
First, note that since we are looking for subsets with small edge expansion, when we cut an edge $ e $, the subset containing $ e^+ $ sustains a loss in its edge expansion. The cause of this deficiency is that the numerator of the edge expansion is added by $ c(e) $ and the denominator is subtracted by $ \omega(T_e) $. With this intuition, for every edge $ e\in E $, define
\begin{equation}
\label{eq:epsilon}
\varepsilon_\xi(e)  \isdef \xi\,
\omega(T_e) + c(e).
\end{equation}
Now, let $ \kappa  $ and $\lambda $ be two nonnegative integers and for every integers $ 1\leq k\leq \kappa$, $0\leq l\leq \lambda $ and vertex $ u\in V(T) $, define  $\mathscr{C}_{\xi}(u,k,l)$ to be the set of all $ k $-subpartitions $\mathcal{A}=\{A_i\}_1^k$ in $ \mathscr{C}_{k}(T_u) $ such that $ u\in A_1 $ and $|R(\mathcal{A},T_u)| \leq l$ and for each $ 2\leq i\leq k $, we have $\phi_T(A_{{i}}) \leq \xi$. For each such subpartition $\mathcal{A}$, let $F_{\mathcal{A}}\isdef \partial A_1\setminus \{e_u\}$. Note that any pair of edges in $F_{\mathcal{A}}$ are incomparable and define,
\[\gamma_\xi(\mathcal{A}) \isdef \displaystyle \sum_{e \in F_{\mathcal{A}}} \varepsilon_\xi(e).\]
We will shortly see that minimizing the edge expansion $ \phi_T(A_1) $, in some sense, is equivalent to minimizing $ \gamma_\xi(\mathcal{A}) $ (see \eqref{eq:gamma}). Thus, define,
\begin{equation} \label{eq:gammadef}
 \Gamma_\xi(u,k,l) \isdef 
\displaystyle\min_{\mathcal{A} \in \mathscr{C}_\xi(u,k,l)}\ \gamma_\xi(\mathcal{A}).
\end{equation}
On the other hand, for every integers $0\leq k\leq \kappa$ and $0\leq l\leq \lambda$ and vertex $ u\in V(T) $, define
 $\mu_{\xi}(u,k,l)$ to be equal to $ 1 $ if there exists a connected $ k $-subpartition  $\mathcal{A}=\{A_i\}_1^k\in \mathscr{C}_k(T_u) $ such that  $\phi_T(\mathcal{A}) \leq \xi$ and $|R(\mathcal{A},T_u)| \leq l$ and it is equal to $ 0 $, otherwise. Note that, although $ A_i$'s are subsets of  $V(T_u) $, $ \phi_T(A_i) $ is computed in the whole tree $ T $.
Also, note that for every vertex $ u\in V(T) $ and integer $ l $, we have
\[ \mu_\xi(u,0,l) =\begin{cases}
1 & \text{ if } |V(T_u)|\leq l,\\
0 & \text{ o.w.}
\end{cases} \]

In fact, our main goal is to compute the parameter $ \mu_{\xi}(r,\kappa,\lambda) $, since evidently the answer to CMSC problem is yes if and only if $ \mu_{\xi}(r,\kappa,\lambda)=1$.
In the sequel, we are going to show that the parameters $\Gamma_\xi(u,k,l) $ and $\mu_{\xi}(u,k,l)$ can be computed recursively in a breath-first scanning of vertices 
towards the root. First, in the following, we explain how one can compute $ \mu_{\xi}(u,k,l) $ recursively in terms of the values $ \Gamma_\xi(u,k,l) $, $ 1\leq k\leq \kappa, 0\leq l\leq \lambda $.
For this, let $ \xi,\kappa, \lambda $ be fixed and given a vertex $ u $, let $ (u_1,\ldots,u_d) $ be an ordering of all of its children. Now, for every integers $ 0\leq k\leq \kappa $, $ 0\leq l\leq \lambda $, define
\begin{equation}
U(1,k,l) \isdef \mu_\xi(u_1,k,l),
\end{equation}
and for every $ 2\leq i\leq d $, define
\begin{equation}\label{eq:recmu}
U(i,k,l) \isdef \begin{cases}
1 & \parbox{8cm}{if there exist $0\leq k'\leq k$  and  $0\leq l'\leq l-1$ such that
 $U(i-1,k',l')= \mu_\xi(u_i,k-k',l-1-l')=1,$} \\
0 & \text{o.w.}
\end{cases}
\end{equation}
In the following lemma, we show how one can use the recursion in \eqref{eq:recmu} to compute the function $ \mu_\xi $.
\begin{lem}\label{lem:mu}
Let $ u $ be a vertex in a rooted tree $ T $, $ \xi\geq 0 $ be a number and $ \lambda\geq 0, \kappa\geq 1 $ be two integers. 
Also, let $ u_1,\ldots, u_d $ be the children of $ u $ in $ T $.
For every integers $ 0\leq k\leq \kappa $, $ 0\leq l\leq \lambda $, $ \mu_\xi(u,k,l)=1 $ if and only if either $ \Gamma_\xi(u,k,l) \leq \xi\, \omega(T_{u}) - c(e_u)  $, or $ U(d,k,l)=1 $.
\end{lem}
\begin{proof}
Suppose that $\mu_{\xi}(u,k,l)=1 $ and let $\mathcal{A}=\{A_i\}_1^k \in \mathscr{C}_{k}(T_u) $ be a connected  subpartition where $ R(\mathcal{A},T_u)\leq l$ and $ \phi_T(\mathcal{A})\leq \xi $. First, assume that $ u\in R(\mathcal{A},T_u)$. Thus, $ \mathcal{A} $ itself can be partitioned into $ d $ connected subpartitions $ \mathcal{A}_1,\ldots, \mathcal{A}_d $ such that $ \mathcal{A}_i\in \mathscr{C}_{k_i}(T_{u_i}) $, for some integers $ k_i $, where $ k_1+\cdots+k_d=k $. Also, let $ l_i= |R(\mathcal{A}_i,T_{u_i})| $. Therefore, by definition $ \mu_{\xi}(u_i,k_i,l_i)=1 $ and $ l_1+\cdots+l_d=|R(\mathcal{A},T_u)\setminus\{u\}|\leq l-1 $. Thus, again by definition $U(d,k,l)=1$.
Next, suppose that $ u\not \in R(\mathcal{A},T_u) $ and so, without loss of generality, assume that $ u\in A_1 $. 
Then,
\begin{align}\label{eq:gamma}
\begin{split}
\phi_T(A_1)\leq \xi &\ \Leftrightarrow \   c(\partial A_1) \leq \xi\, \omega(A_1) \ \Leftrightarrow \
c(e_u)+\sum_{e\in F_{\mathcal{A}}} c(e) \leq \xi \left(\omega(T_{u})-\sum_{e\in F_{\mathcal{A}}} \omega(T_e)\right) \\
&\ \Leftrightarrow \ \gamma_\xi(\mathcal{A}) \leq \xi\, \omega(T_{u})-c(e_u).
\end{split}
\end{align}
Therefore, $ \Gamma_\xi(u,k,l)\leq \xi\, \omega(T_u)-c(e_u) $. This implies that if  $ \mu_\xi(u,k,l)=1 $, then either $ \Gamma_\xi(u,k,l) \leq \xi\, \omega(T_{u}) - c(e_u)  $, or $ U(d,k,l)=1 $.

Now, suppose that $ U(d,k,l)=1 $. Then, there exist integers $ l_1,\ldots,l_d $ and $ k_1,\ldots, k_d $ such that $ \sum_{i=1}^d l_i=l-1 $, $ \sum_{i=1}^d k_i=k $ and $ \mu_\xi(u_i,k_i,l_i)=1 $, for all $ 1\leq i\leq d $.
Thus, for each $ 1\leq i\leq d $, there exists $ \mathcal{A}_i \in \mathscr{C}_{k_i}(T_{u_i})$ such that $ R(\mathcal{A}_i,T_{u_i})\leq l_i $ and $ \phi_T(\mathcal{A}_i)\leq \xi $. Define $ \mathcal{A}=\cup_{i=1}^d \mathcal{A}_i $. Thus, $ \mathcal{A}\in \mathscr{C}_{k}(T_{u}) $ and $ |R(\mathcal{A},T_u)|\leq l $. Hence, $  \mu_{\xi}(u,k,l)=1 $. 

Finally, suppose that $ \Gamma_\xi(u,k,l) \leq \xi\, \omega(T_u) - c(e_u) $. Also, let $ \mathcal{A}\in \mathscr{C}_\xi(u,k,l) $ be a minimizer with $ \gamma_\xi(\mathcal{A})=\Gamma_\xi(u,k,l) $. Then, by definition, for every $ 2\leq i\leq k $, $ \phi_T(A_i)\leq \xi $ and $ |R(\mathcal{A},T_u)|\leq l $ and by \eqref{eq:gamma}, $ \phi_T(A_1)\leq \xi $. Hence, $\mu_{\xi}(u,k,l)=1 $. This completes the proof.
\end{proof}

As we see in Lemma~\ref{lem:mu}, in order to obtain the value of $ \mu_\xi(u,k,l) $, we require to have the value of $ \Gamma_\xi(u,k,l) $.  
In the next step, we show that  given $\xi,\kappa$ and $\lambda$, how one may compute $\Gamma_\xi(u,k,l)$ efficiently for all vertices $u$ and integers $1\leq k\leq \kappa$, $ 0\leq l\leq \lambda $.
For this, let $ \xi,\kappa, \lambda $ be fixed and given a vertex $ u $, let $ (u_1,\ldots,u_d) $ be an ordering of all of its children. Now, for every integers $ 1\leq k\leq \kappa $, $ 0\leq l\leq \lambda $ and $ 1\leq i\leq d $, define

\begin{equation}\label{eq:X}
X_\xi(i,k,l)\isdef \begin{cases}
\min\{\varepsilon_\xi(uu_i), \Gamma_\xi(u_i,k,l)\} &  \text{if } \mu_{\xi}(u_i,k-1,l)=1, \\
\Gamma_\xi(u_i,k,l) & \text{o.w.}
\end{cases}
\end{equation}
Also, define 
\begin{equation}\label{eq:Yinit}
Y_\xi(1,k,l)\isdef X_\xi(1,k,l),
\end{equation}
and for every $ 2\leq i\leq d $, define
\begin{equation}\label{eq:Y}
Y_\xi(i,k,l)\isdef \min\{Y_\xi(i-1,k',l')+X_\xi(i,k+1-k',l-l'):\ 1\leq k'\leq k, 0\leq l'\leq l\}.
\end{equation}
The following lemma shows how to compute the function $ \Gamma_\xi $ using recursion \eqref{eq:Y}.
\begin{lem}\label{lem:Gamma}
Let $ T $ be a rooted tree, $ \xi\geq 0 $ be a number and $ \lambda\geq 0, \kappa\geq 1 $ be two integers. Then, for every vertex $u\in V(T)$ with $d$ children $ (u_1,\ldots, u_d) $ and every integers
$ 1\leq k\leq \kappa  $ and $ 0\leq l \leq \lambda $, we have
\begin{equation}\label{eq:Gamma}
\Gamma_\xi(u,k,l)= Y_\xi(d,k,l).
\end{equation}
\end{lem}
\begin{proof}
We prove the lemma by induction on the number $ d $.
Let $ \mathcal{A}=\{A_i\}_1^k \in \mathscr{C}_\xi(u,k,l) $ be a $ k $-subpartition.
 First, suppose that $ d=1 $. If $ uu_1\in F_{\mathcal{A}} $, then $ \gamma_\xi(\mathcal{A}) =\varepsilon_\xi(uu_1) $ and $ A_2,\ldots,A_k\subseteq V(T_{u_1}) $, so $ \mu_{\xi} (u_1,k-1,l)=1 $. Also, if $ uu_1\not \in F_{\mathcal{A}} $, then $ u_1\in A_1 $ and $ \gamma_\xi(\mathcal{A})\geq \Gamma_{\xi}(u_1,k,l) $. Therefore, $ \Gamma_\xi(u,k,l)= X_\xi(1,k,l)= Y_\xi(1,k,l) $ as in \eqref{eq:X} and \eqref{eq:Yinit}.

Now, suppose that $ d\geq 2 $. Let $ T'\isdef T\setminus T_{u_d} $ and $ T''\isdef T\setminus (\cup_{i=1}^{d-1} T_{u_i}) $ and $\Gamma(k,l)$, $\Gamma'(k,l)$ and $\Gamma''(k,l) $ be the values of $ \Gamma_\xi(u,k,l) $ for the trees $ T $, $ T' $ and $ T'' $, respectively. Also, let $ l'=|R(\mathcal{A},T)\cap V(T')| $, $ l''=l-l' $ and let $ k' $ (resp. $ k'' $) be the number of sets $ A_i $ which intersect $ V(T') $ (resp.  $ V(T'') $). Then, evidently we have $ k'+k''=k+1 $ (note that $ A_1 $ intersects both $ V(T') $ and $ V(T'') $) and $ \gamma_\xi(\mathcal{A})\geq \Gamma'(k',l')+\Gamma''(k'',l'')$.  
Therefore,
$$ \Gamma(k,l)=\min \{\Gamma'(k',l')+\Gamma''(k'',l''): \ k'+k''=k+1, l'+l''=l \}.  $$
On the other hand, by the induction hypothesis, we have 
$\Gamma'(k',l')= Y_\xi(d-1,k',l')$  and  $\Gamma''(k'',l'')=X_\xi(d,k'',l'') $.
Hence, by \eqref{eq:Y}, we have $ \Gamma(k,l) =Y_\xi (d,k,l) $ and we are done.
\end{proof}

\section{The algorithm} \label{sec:alg}
In this section, using Lemmas~\ref{lem:mu} and \ref{lem:Gamma}, we provide an algorithm to solve the CMSC problem for all weighted trees. The cores of the algorithm are two dynamic programmings. The final solution to the problem is given in \pref{alg:main} which scans the vertices in a BFS order towards the root $ r $ and computes recursively the values of $ \Gamma_\xi(u,k,l) $ and $ \mu_{\xi}(u,k,l) $, for $ 1\leq k\leq \kappa  $ and $ 0\leq l\leq \lambda $. The structure of \pref{alg:main} which deploys Algorithms~\ref{alg:Gamma} and \ref{alg:mu} as two subroutines, is as follows.

First, for all leaves $ u $ (vertices with no children), it computes the values of $ \Gamma_\xi(u,k,l) $ and $ \mu_{\xi}(u,k,l) $ (Lines~\ref{init1}-\ref{init2} in \pref{alg:main}). Next, for a vertex $ u $, with children $ (u_1,\ldots, u_d) $, according to \pref{lem:Gamma}, it applies a dynamic programming (\pref{alg:Gamma}) based on the recursion given in Equations~\eqref{eq:Yinit} and \eqref{eq:Y}, to obtain the value of $ \Gamma_\xi(u,k,l) $, assuming the values of $ \mu_{\xi}(u_i,k,l) $ and $ \Gamma_\xi(u_i,k,l) $ are given.
Finally, according to Lemma~\ref{lem:mu}, it applies another dynamic programming (\pref{alg:mu}) based on the recursion given in \eqref{eq:recmu} to obtain the value of $ \mu_{\xi}(u,k,l) $, assuming the values of $ \Gamma_\xi(u,k,l) $ and $ \mu_{\xi}(u_i,k,l) $ are given. The backtracking ends up outputting the value of $ \mu_{\xi}(r,\kappa,\lambda) $ which is equal to $ 1 $ if and only if there exists a connected $ \kappa $-subpartition $ \mathcal{A}$ with $  \phi_T(\mathcal{A}) \leq \xi$ and $ |R(\mathcal{A},T)|\leq \lambda $. This completes the solution. 

\begin{algorithm}[ht]
	\caption{\label{alg:Gamma} \\
		\textbf{Input:} A weighted tree $(T,\omega,c)$, a rational number $\xi$ and integers $\kappa\geq 1$ and $\lambda\geq 0$. A vertex $ u\in V(T) $ with children $ (u_1,\ldots,u_d) $. The numbers $ \mu_{\xi}(u_i,k,l) $ and $ \Gamma_\xi(u_i,k,l) $, for all $ 1\leq i\leq d $, $  1\leq k\leq \kappa $ and $ 0\leq l\leq \lambda $. \\
		\textbf{Output:} The numbers $ \Gamma_\xi(u,k,l) $  for all $  1\leq k\leq \kappa $ and $ 0\leq l\leq \lambda$.	
	}
	\begin{algorithmic}[1]	
\State \label{ln:1-1} Set $\varepsilon_i:= \xi \omega(T_{u_i})+c(uu_i)$, $\forall\, 1\leq i\leq d $;
		
		
	\For{$ i=1 :d $} \label{alg1-l1}
	\For{$ l=0:\lambda $}
	\For{$ k=1:\kappa $}
	\If {$\mu_{\xi}(u_i,k-1,l)=1 $ and $\varepsilon_i \leq \Gamma_\xi(u_i,k,l) $}
	\label{ln:1-5} 
		  \State $X(i,k,l) \isdef \varepsilon_i$;
	  \Else
	      \State $X(i,k,l) \isdef \Gamma_\xi(u_i,k,l)$;
	  \EndIf
	\EndFor 
	\EndFor 
	\EndFor \label{alg1-l2} 
\State \label{ln:1-13} Set $ Y(1,k,l):= X(1,k,l) $, $\forall\, 1\leq k\leq \kappa, 0\leq l\leq \lambda$; 
\For{$ i=2 :d $}\label{alg1-l3}
\For{$ l=0:\lambda $}
\For{$ k=1:\kappa $}
\State $ Y(i,k,l)\isdef +\infty $;
\For{$ l'=0:l$}
\For{$ k'=1:k$}

\State $Y(i,k,l) \isdef \min\{ Y(i,k,l),Y(i-1,k',l')+X(i,k+1-k',l-l') \}$;\label{alg1-l20}


\EndFor 
\EndFor 

\EndFor 
\EndFor 
\EndFor 
\label{alg1-l4} 
\State  \label{ln:1-26} Set $ \Gamma_{\xi}(u,k,l)\isdef Y(d,k,l) $, $\forall\, 1\leq k\leq \kappa, 0\leq l\leq \lambda$; 

\State \Return $\Gamma_{\xi}(u,k,l) $, for all $ 1\leq k\leq \kappa $ and $0\leq l\leq \lambda  $;	
\end{algorithmic}
\end{algorithm}

\begin{algorithm} [ht]
	\caption{\label{alg:mu} \\
		\textbf{Input:} A weighted tree $(T,\omega,c)$, a rational number $\xi$ and integers $\kappa\geq 1$ and $\lambda\geq 0$. A vertex $ u\in V(T) $ with children $ (u_1,\ldots,u_d) $. The numbers $ \Gamma_\xi(u,k,l) $, for all $  1\leq k\leq \kappa $ and $ 0\leq l\leq \lambda $. The numbers $ \mu_{\xi}(u_i,k,l) $ for all $ 1\leq i\leq d $, $ 1\leq k\leq \kappa $ and  $ 0\leq l\leq \lambda $. \\
		\textbf{Output:} The numbers $ \mu_{\xi}(u,k,l) $ for all $ 1\leq k\leq \kappa $ and $ 0\leq l\leq \lambda$.	
	}
	
	\begin{algorithmic}[1]	
		
		
		\State Let $ e_u $ be the parent edge of $ u $ and $ u_1,\ldots,u_d $ be the children of $ u $.
		\For{$ l=0:\lambda $} \label{alg2-l1}
		\For{$ k=1:\kappa $}
		\State $ \mu_\xi(u,k,l):= 0 $;
		\If {$ \Gamma_\xi(u,k,l)\leq \xi \omega(T_u)-c(e_u) $} \label{ln:2-5}
		\State $\mu_\xi(u,k,l):= 1 $;
		\EndIf
		\EndFor
		\EndFor \label{alg2-l2}
		
		\State \label{ln:2-10} Define $U(1,k,l) \isdef \mu_\xi(u_1,k,l)$, for every $ 0\leq k\leq \kappa, 0\leq l\leq \lambda $;
		
		\For{$ i=2:d $} \label{alg2-l3}
		\For{$ l=0:\lambda $}
		\For{ $ k=0: \kappa $}
		\If {$\mu_\xi(u,k,l)= 0 $}
		\State $ U(i,k,l):=0 $;
		\For{$ l'=0:l-1 $}
		\For{ $ k'=0: k $}
		\If {$ U(i-1,k',l')= \mu_\xi(u_i,k-k',l-1-l')=1 $} \label{ln:2-18}
		\State $ U(i,k,l):=1 $ and go to Line~\ref{line:endif};
		\EndIf 
		\EndFor
		\EndFor
		\label{line:endif}\EndIf 
		\EndFor
		\EndFor
		\EndFor \label{alg2-l4}
		\State \label{ln:2-27} For every $ 0\leq k\leq \kappa $ and $ 0\leq l\leq \lambda  $, if $ U(d,k,l)=1 $, then set $ \mu_\xi(u,k,l):=1 $.

		%
		\State \Return $\mu_{\xi}(u,k,l) $, for all $0\leq l\leq \lambda  $;	
	\end{algorithmic}
\end{algorithm}

\begin{algorithm} [ht]
	\caption{\label{alg:main} \\
		\textbf{Input:} A weighted tree $(T,\omega,c)$, a rational number $\xi$ and integers $\kappa\geq 1$ and $\lambda\geq 0$. \\
		\textbf{Output:} Decide if there exists $ \mathcal{A}\in \mathscr{C}_k(T) $ where $ \phi_T(\mathcal{A})\leq \xi $ and $ |R(\mathcal{A},T)|\leq \lambda $?
	}
	
	\begin{algorithmic}[1]	
		\State Root $T$ with an arbitrary node $r$;
		\State Order all nodes in BFS order with respect to $r$, as $v_{1},...,v_{n} = r$;	
		\State Set $i:=1$;
		\While{ $i\leq n$}
		\State Let $ u:=v_i $ and $ e_u $ be the parent edge of $ u $ and $ u_1,\ldots,u_d $ be the children of $ u $.
		\State Initialize $ \mu_\xi(u,k,l):=0 $, $\forall\ 0\leq k\leq \kappa, 0\leq l\leq \lambda$;
		\State Initialize $ \Gamma_\xi(u,k,l):=+\infty $, $\forall\  1\leq k\leq \kappa, 0\leq l\leq \lambda $;
		\If {$ d=0 $} \label{init1}
		\State Set $ \mu_\xi(u,0,l):=1, \forall\ 1\leq l\leq \lambda$.
		\If {$ c(e_u)\leq \xi \omega(u) $}  
		\State Set $ \mu_{\xi}(u,1,l):=1,\forall\ 0\leq l\leq \lambda$. 
		\EndIf
		\State Set $ \Gamma_\xi(u,1,l):=0 $, $\forall \ 0\leq l\leq \lambda$; \label{init2}
		\Else
		\State Using \pref{alg:Gamma}, find the value of $ \Gamma_\xi(u,k,l) $, for all $ 0\leq l\leq \lambda $ and $ 1\leq k\leq \kappa $; 
		\State Using \pref{alg:mu}, find the value of $\mu_{\xi}(u,k,l) $, for all $ 0\leq l\leq \lambda $ and $ 0\leq k\leq \kappa $;
		\EndIf
		\State \label{ln:endwhile} $ i\leftarrow i+1 $;		
		\EndWhile
		\State If $ \mu_{\xi}(u,\kappa,\lambda)=1$, then \Return Yes. Otherwise, \Return No.
	\end{algorithmic}
\end{algorithm}


\subsection{Time complexity}
The time complexity of the provided algorithms can be computed as follows. In \pref{alg:Gamma}, Lines~\ref{alg1-l1}-\ref{alg1-l2} can be done in $ O(d(\lambda+1)\kappa) $. Also, Lines~\ref{alg1-l3}-\ref{alg1-l4} can be performed in $ O(d(\lambda+1)^2 \kappa^2) $. 
In \pref{alg:mu}, Lines~\ref{alg2-l1}-\ref{alg2-l2} run in $ O((\lambda+1)\kappa) $ and Lines~\ref{alg2-l3}-\ref{alg2-l4} run in $ O(d(\lambda+1)^2 \kappa^2) $. Hence, the runtime of \pref{alg:main} is in $ O((\lambda+1)^2\kappa^2 n) $. Since in real applications, the values of $ \kappa $ and $ \lambda $ are mostly much smaller than $ n $, we can assume that the algorithm runs in linear time with respect to the number of nodes.

\subsection{Constructing the optimal subpartition}
Now, we show that during the execution of Algorithm \ref{alg:main}, how  one can construct a subpartition $ \mathcal{A}\in \mathscr{C}_\kappa(T) $ with $ \phi_T(\mathcal{A})\leq \xi $ and $ |R(\mathcal{A},T)|\leq \lambda $ (if there exists).
Let $ \xi,\kappa $ and $ \lambda $ be fixed and for every vertex $ u\in V(T) $ and $ 1\leq k\leq \kappa $ and $ 0\leq l\leq \lambda $, if $ \mu_\xi(u,k,l)=1 $, then let $ \mathcal{A}_\mu=\mathcal{A}_\mu(u,k,l) $ be a $ k-$subpartition in $ \mathscr{C}_k(T_u)  $ such that $ \phi_T(\mathcal{A}_\mu)\leq \xi$ and $ |R(\mathcal{A}_\mu,T_u)|\leq l $. Also, if  $ \mu_\xi(u,k,l)=0 $, let $ \mathcal{A}_\mu=\mathcal{A}_\mu(u,k,l):= \emptyset $. Then, the subpartition $ \mathcal{A}_\mu(r,\kappa,\lambda) $ is what we are looking for. Also, let $ \mathcal{A}_\Gamma=\mathcal{A}_\Gamma(u,k,l) $ be a subpartition in $ \mathscr{C}_\xi(u,k,l) $ which minimizes \eqref{eq:gammadef}.

Now, let $ u $ be a vertex with children $ (u_1,\ldots,u_d) $. First, according to Algorithm~\ref{alg:mu} and  assuming that we have all the subpartitions $\mathcal{A}_\Gamma(u,k,l)$ and $ \mathcal{A}_\mu(u_i,k,l) $, we explain how to obtain $ \mathcal{A}_\mu(u,k,l) $. For this, throughout the execution of Algorithm~\ref{alg:mu}, in Line~\ref{ln:2-5}, if $ \Gamma_\xi(u,k,l)\leq \xi \omega(T_u)-c(e_u) $, then set $ \mathcal{A}_\mu(u,k,l)\isdef \mathcal{A}_\Gamma(u,k,l) $, otherwise set $ \mathcal{A}_\mu(u,k,l)\isdef \emptyset $. Also, in Line~\ref{ln:2-10}, if $ \mu_\xi(u_1,k,l)=1 $, then set $ \mathcal{U}(k,l)\isdef \mathcal{A}_\mu(u_1,k,l) $ and in Line~\ref{ln:2-18}, if $ U(i-1,k',l')= \mu_\xi(u_i,k-k',l-1-l')=1 $, then set $ \mathcal{U}(k,l)\isdef \mathcal{U}(k',l')\cup \mathcal{A}_\mu(u_i,k-k',l-1-l') $.  Finally, in Line~\ref{ln:2-27}, if $ U(d,k,l)= 1 $, then set $ \mathcal{A}_\mu(u,k,l)\isdef \mathcal{U}(k,l) $.

Next, according to \pref{alg:Gamma} and assuming that we have all the subpartitions $ \mathcal{A}_\mu(u_i,k,l) $ and $ \mathcal{A}_\Gamma(u_i,k,l) $, we explain how to obtain $ \mathcal{A}_\Gamma(u,k,l) $.  
First, throughout the execution of \pref{alg:Gamma}, in Line~\ref{ln:1-5}, if $\mu_{\xi}(u_i,k-1,l)=1 $ and $\varepsilon_i \leq \Gamma_\xi(u_i,k,l) $, then set $ \mathcal{X}(i,k,l)\isdef \mathcal{A}_\mu(u_i,k-1,l)\cup \{\{u\}\}$, otherwise let $ \mathcal{X}(i,k,l)$ be the subpartition obtained from $ \mathcal{A}_\Gamma(u_i,k,l) $ by adding the vertex $ u $ to the set containing $ u_1 $. 
Also, in Line~\ref{ln:1-13}, set $ \mathcal{Y}(k,l)\isdef \mathcal{X}(1,k,l) $. Next, in Line~\ref{alg1-l20}, if $Y(i-1,k',l')+X(i,k+1-k',l-l')\leq Y(i,k,l)$, then let $ \mathcal{Y}(k,l) $ be obtained from the disjoint union of $ \mathcal{Y}(k',l') $ and $ \mathcal{X}(i,k+1-k',l-l') $ by merging two sets containing the vertex $ u $.
Finally, in Line~\ref{ln:1-26}, set $ \mathcal{A}_\Gamma(u,k,l)\isdef \mathcal{Y}(k,l) $.


\section{Towards more extensions} \label{sec:extensions}

In this section, we show that our presented scheme can be generalized to solve the following more realizable problems  efficiently:

\begin{enumerate}

\item Solving CMSC problem on trees with potentials.

\item Solving CMSC problem on forests.

\item Solving the following semi-supervised problem: Given a weighted graph $G=(V,E,\omega,c)$ (not necessarily a  forest), two disjoint subsets $S_1,S_2\subseteq V$, rational number $ \xi $ and integers $ \kappa, \lambda $, such that the induced subgraph of $ G $ on $ V(G)\setminus S_1 $ is a forest. Does there exist a connected subpartition $ \mathcal{A}\in \mathscr{D}_\kappa(V) $ such that $ \phi_G(\mathcal{A})\leq \xi $, $ |R(\mathcal{A},G)|\leq \lambda $, $ S_1\subseteq R(\mathcal{A},G) $ and $ S_2 \cap R(\mathcal{A},G)=\emptyset $?
\end{enumerate}

In the following, we elaborate on the modifications that should be made to tackle the above settings. 
\begin{enumerate}
\item In the setting of trees with potentials, each vertex $v \in V(T)$ is endowed with a potential weight, say $p(v)$, which is a nonnegative number and the goal is to determine whether there exists a connected $k$-subpartition $\mathcal{A}=\{A_{{i}}\}^k_1 \in \mathscr{D}_k(V)$ such that
		$$\displaystyle{\phi_T(\mathcal{A})=\max_{1 \leq i \leq k}}\left\{ \phi_T(A_i)=\frac{ 
			c(\partial A_{{i}}) + p(A_i)}{\omega(A_{{i}})}\right\} \leq \xi$$ 
		and $|R(\mathcal{A},T)|\leq \lambda$. We can extend our method to solve this problem using Algorithm \ref{alg:main}. First, for each edge $e \in E(T)$, amend the definition of $ \varepsilon_\xi(e) $ in \eqref{eq:epsilon} as follows 
			$$\varepsilon_\xi(e)  \isdef \xi\,
 \omega(T_e) + c(e) - p(T_e).$$
Also, define the functions $ \mu_\xi $ and $ \Gamma_\xi $ analogously. Next, with a similar argument as in \pref{lem:mu}, one may prove that $ \mu_\xi(u,k,l)=1 $ if and only if either $ \Gamma_\xi(u,k,l) \leq \xi\, \omega(T_{u}) - c(e_u) - p(T_u)   $, or $ U(d,k,l)=1 $.
Moreover, \pref{lem:Gamma} is still valid. So, we should just change Line~\ref{ln:1-1} in \pref{alg:Gamma} and Line~\ref{ln:2-5} in \pref{alg:mu}, accordingly and then \pref{alg:main} works for the new setting. 

\item  Suppose that the forest $ F $ consists of $c$ disjoint trees $T_1,\ldots,T_c$ rooted at $r_1,\ldots,r_c$ respectively. Also, let $ \xi,\kappa, \lambda $ be fixed. First, using \pref{alg:main}, compute the value of $ \mu_\xi(r_i,k,l) $,  for every integers $ 0\leq l\leq \lambda $, $0\leq k \leq \kappa$ and $ 1\leq i\leq c $. Also,  define
\begin{equation*}
Z_\xi(1,k,l) = \mu_{\xi}(r_1,k,l).
\end{equation*}

The following recursion helps us to solve the problem on $ F $. For every $ 2\leq i\leq c $, define

\begin{equation*}
Z(i,k,l) \isdef \begin{cases}
1 & \text{if there exists } 0\leq k'\leq k \text{ and } 0\leq l'\leq l, \text{ s.t. }\\
& Z(i-1,k',l')= \mu_\xi(r_i,k-k',l-l')=1, \\
0 & \text{o.w.}
\end{cases}
\end{equation*}

Then, the solution to CMSC problem is yes if and only if $Z_\xi(c,\kappa,\lambda) = 1$. Furhermore, One may easily extend this recursion to solve the corresponding problem on forests with potentials.

\item In this setting some vertices should be or should not be in the residue set. The problem can be solved in the following steps:

\begin{itemize}
\item[-] First, for each vertex $v \in V(G)$, define a potential as follows
\[p(v)\isdef \sum_{e\in E(\{v\},S_1)} c(e). \]
 \item[-] Now, let $ F $ be a forest obtained from $ G $ by deleting all vertices in $S_1$. Also, let $ \lambda'\isdef \lambda-|S_1| $.
\item[-] If $ S_2 $ is empty, then the solution can be obtained by performing the method given in 2 on the forest $ F $ with the potential weight $ p $ and the numbers $ \xi,\kappa, \lambda' $. If $ S_2 $ is non-empty, we have to make the following additional modifications to handle the problem.
\end{itemize}

Suppose that $ T $ is a tree and $ S\subseteq V(T) $ is a subset of vertices. Also, numbers $ \xi,\kappa,\lambda $ are given. We are looking for a connected subpartition $ \mathcal{A}\in \mathscr{D}_\kappa(V)  $ such that $ \phi_T(\mathcal{A})\leq \xi $, $ |R(\mathcal{A},T)|\leq \lambda $ and $ S\cap R(\mathcal{A},T)=\emptyset $. Note that \pref{lem:Gamma} is still valid in this setting. However, in the computation of $\mu_\xi(u,k,l)$, for each $u\in V(T)$, in \pref{lem:mu}, if $u \in S$, then $ u $ is not allowed to be in the residue set. So, the value of $\mu_\xi(u,k,l)$ is equal to $1$ if and only if $ \Gamma_\xi(u,k,l) \leq \xi\, \omega(T_{u}) - c(e_u)$. Thus, with a similar proof as in \pref{lem:mu}, we can prove that 
\begin{equation*}
 \mu_\xi(u,k,l) =
 \begin{cases}
1 & \text{if $u \in S$ and  $  \Gamma_\xi(u,k,l) \leq \xi\, \omega(T_{u}) - c(e_u), $} \\
1 & \text{if $u \not\in S$ and  either $  \Gamma_\xi(u,k,l) \leq \xi \omega(T_{u}) - c(e_u)  $ or $ U(d,k,l)=1, $}\\
0 & \text{o.w.}
\end{cases}
\end{equation*}

Then, \pref{alg:mu} can be modified accordingly to compute the value of $ \mu_\xi(u,k,l) $.

\end{enumerate}

\section{Concluding remarks and future work}
In this paper, a multi-way sparsest cut problem has been investigated for weighted trees and it was shown that although the problem is $ NP- $complete for trees, it becomes tractable when the search space is confined to connected subdomains. One of the strengths of the method is that it has a control on the number of outliers and can manage semi-supervised settings when some data points are forced or forbidden to be outlier.
Besides the theoretical importance of the sparsest cut problem, when our method is applied to the minimum spanning tree, it can steer several applications in both unsupervised and semi-supervised clustering. 

One may also consider an analogous problem when we are seeking for a subpartition minimizing ``the average''  (instead of the maximum) of the edge expansions of the parts (e.g. as in \cite{shi2000normalized}). This objective function is more sensitive and exquisite and are more likely to produce high-quality clustering results. Nevertheless, the problem unfortunately turns out to be $ NP- $complete on trees even when the search space is restricted to connected subpartitions (or partitions) \cite{dam}. Finding a good approximation algorithm for this problem is an interesting and challenging task that can be the purpose of future work in this line of research.  \\

\textbf{Acknowledgment.} We would like to express our sincere thanks to Amir Daneshgar whose valuable comments were crucial in preparing and improving the present article.

%
\providecommand{\bysame}{\leavevmode\hbox to3em{\hrulefill}\thinspace}
\providecommand{\MR}{\relax\ifhmode\unskip\space\fi MR }
\providecommand{\MRhref}[2]{%
	\href{http://www.ams.org/mathscinet-getitem?mr=#1}{#2}
}
\providecommand{\href}[2]{#2}

\end{document}